\documentclass[12pt]{article}
\usepackage[margin=1in]{geometry}
\usepackage{amsmath}
\usepackage[english]{babel}
\usepackage[utf8]{inputenc}
\usepackage{graphicx}
\usepackage[usenames]{color}
\usepackage{float}
\usepackage{hyperref}
\usepackage{amsmath}
\usepackage{amsthm}
\usepackage{ulem}
\usepackage{natbib}

\usepackage{lineno}
\usepackage{setspace}
\usepackage{authblk}
\usepackage{url}

\setcitestyle{authoryear,open={(},close={)}}

\providecommand{\keywords}[1]
{
  \small	
  \textbf{\textit{Keywords---}} #1
}

\title{ \textbf{Hidden Markov and semi-Markov models} \\
When and why are these models useful for classifying states in time series data?
	}

\author[1,2]{Sofia Ruiz-Suarez}
\author[3,4]{Vianey Leos-Barajas}
\author[1]{Juan Manuel Morales}

\affil[1]{INIBIOMA (CONICET-Universidad Nacional del Comahue), Quintral 1250, Bariloche,
Rio Negro, Argentina}
\affil[2]{Universidad de Rosario, Facultad de Ciencias Econ\'{o}micas, Bv. Oro\~{n}o 1261, Rosario,Argentina}
\affil[3]{Department of Statistical Sciences, University of Toronto, 700 University Ave, Toronto, ON, M5G 1Z5, Canada}
\affil[4]{School of the Environment, University of Toronto, 33 Wilcocks St, Toronto, ON, M5S 3E8, Canada}

\date{}

\usepackage{fancyhdr}

\usepackage{setspace}
\spacing{1.5}
\begin{document}
\small
\thispagestyle{empty}
\pagenumbering{gobble}
\maketitle
\newpage
\begin{abstract}

Hidden Markov models (HMMs) and their extensions have proven to be powerful tools for classification of observations that stem from systems with temporal dependence as they take into account that observations close in time are likely generated from the same state (i.e.\ class). When information on the classes of the observations is available in advanced, supervised methods can be applied.
In this paper, we provide details for the implementation of four models for classification in a supervised learning context: HMMs, hidden semi-Markov models (HSMMs), autoregressive-HMMs, and autoregressive-HSMMs. Using simulations, we study the classification performance under various degrees of model misspecification to characterize when it would be important to extend a basic HMM to an HSMM. As an application of these techniques we use the models to classify accelerometer data from Merino sheep to distinguish between four different behaviors of interest. In particular in the field of movement ecology, collection of fine-scale animal movement data over time to identify behavioral states has become ubiquitous, necessitating models that can account for the dependence structure in the data. We demonstrate that when the aim is to conduct classification, various degrees of model misspecification of the proposed model may not impede good classification performance unless there is high overlap between the state-dependent distributions, that is, unless the observation distributions of the different states are difficult to differentiate. 

\end{abstract}

\section{Introduction}

The aim in the classification problem is either to allocate data to different groups of interest, or to discover sets of patterns that reflect important dynamics of interest. For systems that exhibit temporal dependence, the goal is to correctly assign different segments of data to a finite set of groups, taking into account that observations near each other (in time) are likely to correspond to the same group. For example, in computational linguistics the goal is to identify words and phrases in spoken language, i.e speech recognition \citep{Juang, Deng}; in meteorology, weather change can be monitored analyzing sequential measures from meteorological radars \citep{Rico-Ramirez, ruiz-suarez_sofia_tecnicas_2019}; in neurophysiology, different brain activities can be distinguished assessing physiological variables such as heart rate or electrocardiograms through time\citep{cheng,Inan}; and in ecology, a set of biologically relevant animal behaviors can be identified from observed acceleration data \citep{nathan_using_2012,leos_barajas_analysis_2017}.

In ecology, a lot of attention has been paid to understand and model animal movement \citep{mevin_b._hooten_animal_2017} as it plays important roles in the fitness and evolution of species\citep{nathan2008}, the structuring of populations and communities \citep{morales_building_2010}, and responses to environmental change \citep{Jonsen2016}. In order to understand how animals move as they respond to internal conditions and external environments, it is essential to be able to distinguish between a set of biologically relevant behaviors. Tri-axial acceleration (ACC) data is now commonly recorded using biologging devices,  allowing researchers to investigate the performance, energy expenditure, and behavior of free-living animals \citep{Williams2020}. These devices measure the change in speed over time in three directions, which can be described relative to the body of the individual. With this information it is possible to identify different activity patterns to then distinguish between different behaviors \citep{wilson2008,williams_can_2015}.

There are several techniques used to solve classification problems \citep{trevor_hastie_elements_nodate}: classification trees, logistic regression, discriminant analysis, neural networks, boosted regression trees, random forests, deep learning methods, nearest neighbors, support vector machines, etc. Many of these techniques have been proposed to classify animal behavior from accelerometer data. For example, \citep{nathan_using_2012} 
identified behavioral modes of griffon vultures using a selection of nonlinear and decision tree methods; \citep{carroll_supervised_2014} trained a support vector machine to classify penguin behavior as either ‘prey handling’ or ‘swimming’; \citep{williams_can_2015} examined the ability of k-nearest neighbour algorithms to distinguish between flight behaviors of Andean condors and griffon vultures; \citep{chakravarty_novel_2019} developed an hybrid model combining biomechanical features and support vector machines to identify between four possible behaviors of Kalahari Meerkats; and \citep{studd_behavioral_2019} used a random forest algorithm and a manually created decision tree to associate observed behaviors of North American red squirrels with logger recorded acceleration and temperature. All these methods assume that the observations are independent, yet time series data presents sequential correlation. This characteristic of the data can be exploited to improve the prediction accuracy of the classifiers \citep{Geurts,dietterich_machine_2002}. 

The class of hidden Markov models (HMMs) \citep{zucchini_hidden_2017,fruhwirth-schnatter_handbook_2019} provide an intuitive framework for the classification of systems with temporal dependence that experience changes in patterns over time connected with underlying shifts in a latent process of interest. HMMs assume that the observation(s) at each point in time are the result of the unknown (hidden) ‘state’ (or class) of the system. The state of the system is assumed to change over time according to the Markov property, i.e. the conditional probability distribution of future states (conditional on both past and present states) depends only upon the present state. Thus, an HMM is defined as a doubly stochastic process composed of an observation process $X_t$ and a (latent) state process $C_t$, where the state process is usually taken to be a first-order Markov chain and the distribution of $X_t$ depends only on the current state $C_t$ and not on previous state or observations.
%$X_i \perp X_j | C_t$
As such, HMMs provide a clear manner in which to do classifications of processes that evolve over time, as they take into account the sequential dependence present in the data and the temporal structure of the consecutive states. 

HMMs can be used for state prediction (supervised approach) or to make inferences about drivers of behavior (unsupervised approach). In the supervised learning context information on the classes of the observations is available in advanced (there is a pre labelled data set), and the number of possible states are known by the user. Alternatively, in the unsupervised learning context there are no labelled data and the number of states is not pre-defined. HMMs have successfully been implemented to classify accelerometer data: \citep{aiguang_li_physical_2010} utilized an HMM to recognize human physical activities using two-second summary features from tri-axial acceleration data; \citep{wang_recognizing_2011} presented an HMM to recognize six human daily activities from sensor signals collected from a single waist-worn tri-axial accelerometer; \citep{leos_barajas_analysis_2017} provided the details necessary to implement and assess an HMM in both the supervised and unsupervised learning context and outlined two applications to marine and aerial systems (shark and eagle) using unsupervised learning.

HMMs are models mostly formulated in discrete time, and as the state process is assumed to be a first order Markov chain, the number of consecutive time points that the system spends in a given state (sojourn time), follows a geometric distribution \citep{zucchini_hidden_2017,LANGROCK2011715}. The popularity of HMMs stems, in part, from the ease with which they can be extended to accommodate other forms of dependence and structures. For instance, the state-duration distribution can be generalized so that the underlying stochastic process is a semi-Markov chain. These models are called hidden semi-Markov models (HSMM) \citep{yu_hidden_2010}, and by being more flexible they allow more realism, improving the classifications when the sojourn time distributions are far from being geometric. HSMMs have been successfully applied in many areas, mostly in speech recognition \citep{chen,Hung-Yan,Hieronymus,Oura}, but also to classify human activities of daily living \citep{Duong,CHUNG20081572}, handwriting \citep{Kashi,Benouareth2007}, and human genes in DNA \citep{Kulp1996}. HSMMs have also been used for language identification \citep{Marcheret}, the prediction of protein structure \citep{Aydin2006}, event recognition in videos \citep{Hongen}, financial time series modelling \citep{BULLA20062192}, classification of music \citep{XiaoBingLiu}, and remote sensing \citep{Pieczynski}, among others.

Classical HMMs and HSMMs assume independence between the observations conditional on the state, but sometimes data is taken at high temporal resolution making this assumption unrealistic which can affect the performance of the classifiers. In such cases, autoregressive structure can be considered to model the observations of each state, leading to autorregresive HMMs and HSMMs (AR(p)-HMM and AR(p)-HSMM) \citep{xu_regularized_2020}, also commonly known as Markov-switching models. Even though these models have been extensively studied, as far as we are aware, these HMM extensions have not been used to classify acceleration data into animal behaviors. 

In this paper, we give an overview of HMM, HSMM, AR(p)-HMM and AR(p)-HSMM; we explain the structure of these models in the supervised classification context, how they can be relaxed and their differences in modelling, inference, estimation, and prediction. We present their formulations and derivations as a compilation of the literature published on this topic extending them to HSMMs with or without autorregresive structure. We then study how these models perform classifications under different scenarios in order to characterize when it would be important to extend an HMM to an HSMM. Finally, we use them to classify accelerometer data from domestic Merino sheep (\textit{Ovis aries}) to distinguish between four different behavioral states.  

\section{Models}
\subsection{Hidden Markov Model (HMM)}
HMMs are composed of two layers: a latent process $\{C_t\}_{t=1}^T$, commonly referred to as the state process, satisfying the Markov property; and an observable state-dependent process $\{X_t\}_{t=1}^T$, where 
$P(X_t|X_1,X_2,\dots X_{t-1},C_t)=P(X_t|C_t)$ 
%%$X_i \perp X_j | \boldsymbol{C}$ 
(Figure \ref{Diagram} (a)). 
At each point in time $t$, $C_t$ is assumed to take on one of $J$ possible values, $C_t \in \{1, 2, \ldots, J\}$, where $J$ denotes the number of states. 

\begin{figure}
 \centering
 \includegraphics[width=0.7\textwidth]{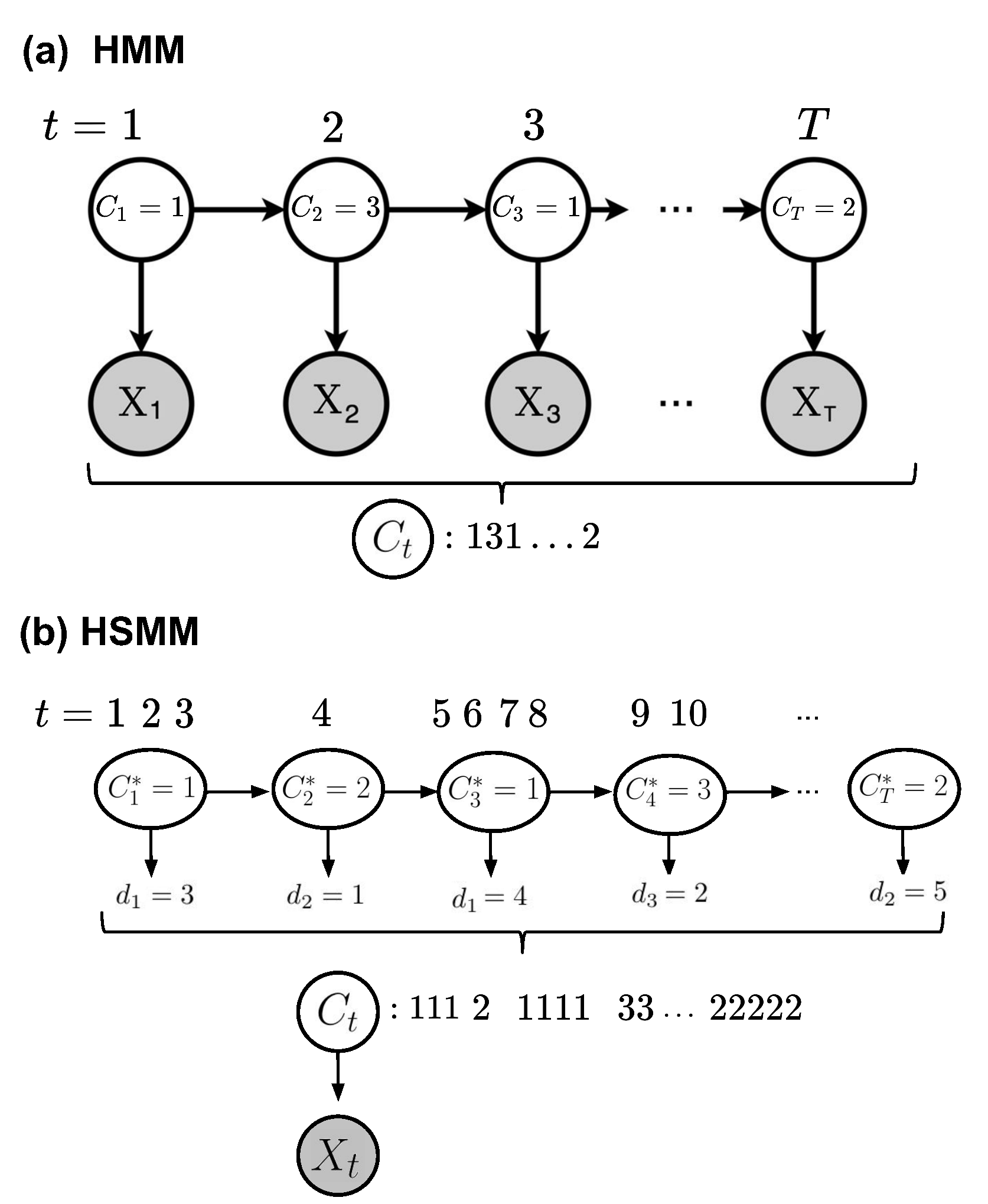}
 \caption{Diagrams of model structure (a)HMM : $C_t$ denotes the latent Markov process and $X_t$ denotes the observation process whose distribution depends on the state $C_t$. (b) HSMM example:  $C_t$ denotes the latent semi-Markov process and $X_t$ the observation process. $C_t^*$ indicates the Markovian process of the non absorbing times (that is, state at time $t$ is equal to the state at time $t-1$), and for each of them $d_t$ gives the values of the sample sojourn times.}
 \label{Diagram}
 \end{figure}
 
The dynamics of the state process $\boldsymbol{C}$ is governed by a $J \times J$ transition probability matrix (t.p.m), $\boldsymbol{\Gamma}$, with entries $\gamma_{ij} = \text{Pr}(C_{t+1} = j|C_t = i)$, for $i,j \in \{1, \ldots, J\}$. The initial probabilities are denoted by $\delta_{i}=\operatorname{Pr}(C_1=i)$,
$i,j = 1,\dots,J$.

By definition, the assumption of a first-order Markov chain structure for the state process implies that the number of consecutive time points that the process spends in a given state before switching, i.e. the sojourn time, follows a geometric distribution with parameter $\gamma_{ii}$. As a consequence, the most likely sojourn time for every state of an HMM is one, and the probability of remaining in a given state decays geometrically. To complete the definition of an HMM, we denote the state-dependent distributions of the observation process (univariate or multivariate) as $f_j(x) = \operatorname{Pr}(X_t=x|C_t=j)$, for $j=1,\dots J$. The state-dependent distributions can be discrete or continuous and it is common to assume the same parametric distribution across all states. 
Therefore an HMM is defined by the pair of stochastic processes $\{X_t,C_t\}$, the state dependent distributions $f_1(x),\dots,f_J(x)$, the transition probability matrix $\boldsymbol{\Gamma}$, and the initial probability vector $\boldsymbol{\delta}=(\delta_1,\dots,\delta_J)$.  

\subsection{Hidden Semi-Markov Model (HSMM)}

One of the limitations of a basic HMM is the assumption that the sojourn time follows a geometric distribution. In fact, much work has been done to extend the basic HMM and explicitly model the sojourn time by other discrete-valued distributions, resulting in the class of hidden semi-Markov models (HSMM) \citep{yu_hidden_2010}. Like an HMM, an HSMM is a doubly stochastic process composed of two layers: a state process $\{C_t\}_{t=1}^T$ and an observable state-dependent process $\{X_t\}_{t=1}^T$. For an HSMM, however, the state process is now assumed to be a semi-Markov chain where the sojourn time in a given state is taken to follow any valid discrete distribution (e.g.\ Poisson, negative binomial). We provide  definitions and notation of an HSMM in the remainder of the section. 

In order to formalize the structure of the model, we first divide the times $t$ into two disjoint categories depending on state $C_t$: 
\begin{enumerate}
    \item \textbf{Non absorbing time}(NAT): The state at time $t$ is different than the state at time $t-1$. ($C_t\neq C_{t-1}$) 
    \item \textbf{Absorbing time}(AT): The state at time $t$ is equal to the state at time $t-1$. ($C
    _t=C_{t-1}$). 
 \end{enumerate}

The subsequence of states of non absorbing times is a first-order Markov process, while the sequence of $C_t$ values is a finite-state semi-Markov chain.
For simplicity, in this paper we assume that state switches occurred at time $t=1$ and $t=T$, such that $C_1 \neq C_0$ and $C_T \neq C_{T+1}$. A HSMM with that simplifying assumption is called a right censored HSMM \citep{guedon}.

For each state $i$ of the non absorbing times, the sojourn time is a realization of the corresponding sojourn-time distribution $d_i$. 
Suppose that $t^*$ is a non absorbing time of state $i$ then the probability that the state $i$ lasts $u$ steps, 
$d_i(u)=\operatorname{Pr}(C_{t^*+u+1}\neq i \wedge C_{t^*+v}=i, v=1,...u|C_{t^*}=i), \quad u\geq 1$. 
For an HSMM, the t.p.m contains the conditional transition probabilities $(\Gamma)_{ij}=\gamma_{ij}=\operatorname{Pr}(C_{t+1}= j|C_t=i , C_{t+1}\neq i)\quad  i\neq j, \quad j = 1,\dots,J$ with $\gamma_{ii}=0 \quad \forall i=1\dots,J$. The probabilities of self-transitions are determined by $d_i$-sojourn distributions.

In Figure \ref{Diagram}(b) an example of a $3$-state HSMM ($C_t \in \{1,2,3\}$) is represented. 
The Markovian layer $C_t^*$ governs the change of states and gives rise to a run of $d$ values all equal to  $C_t^*$.  These $d$ values are realizations from the sojourn time distributions $d_i$ ($i=1,2,3$). In this example $C_1^*=1$ and $d_1=3$,  
indicating that the semi-Markov process starts with three ones, i.e. $C_1=1$, $C_2=1$ and $C_3=1$. Then $C_2^*=2$ and $d_2=1$, indicating that the sequence follows with a single 2, $C_4=2$ . Next  $C_3^*=1$ and $d_1=4$, so four ones are added to the sequence, $C_5=1$, $C_6=1$,$C_7=1$, $C_8=1$  and so on. The observation process
is akin to an HMM in that the observations $X_t$ are random variables that depend on the values of $C_t$. Therefore a HSMM can be defined by the pair of stochastic process $\{X_t,C_t\}$, $X_t$ dependent on a semi-Markov chain $C_t$, the state dependent distributions $f_1(x),\dots,f_J(x)$, the sojourn-time distributions $d_1(u).\dots,d_J(u)$, the t.p.m $\boldsymbol{\Gamma}$, and the initial state distribution $\boldsymbol{\delta}=(\delta_1,\dots,\delta_J)$. 

For the remainder of this article, we use the following notation. The observed sequence of length $l$, $(X_t,X_{t+1}\dots,X_{t+l})$ is denoted by $\boldsymbol{X}_{t:t+l}$, and  the state sequence $(C_t,C_{t+1},\dots C_{t+l})$ is denoted by $\boldsymbol{C}_{t:t+l}$. We use $\boldsymbol{C}_{[t_1;t_2]}=i$ to denote that the vector of state variables $\{C_{t_1} =i , C_{t_1 + 1}=i, \ldots, C_{t_2}=i\}$ takes on the value $i$ with state changes occurring at $t_1$ and $t_{2}+1$, $C_{t_1-1}\neq i$ and $C_{t_2+1}\neq i$.
We let $\boldsymbol{C}_{[t_1;t_2}=i$ denote similarly that the sequence of state variables $\{C_{t_1}=i, \ldots, C_{t_2}=i\}$ take on the value $i$ with a state change occurring at time $t_1$, i.e.\ $C_{t_1-1} \neq i$, but allow for $C_{t_2+1} \in \{1, \ldots, J\}$. Similarly $\boldsymbol{C}_{t_1;t_2]}=i$ denotes a state change at time $t_2 + 1$ but allows for $C_{t_1-1} \in \{1, \ldots, J\}$. By $C_{[t_1=i}$ we mean that at time $t_1$ the system switched from some other state to $i$, and  by $C_{t_1]=i}$ that at time $t_1$ the state will end and transit to some other state.
Finally, we let $\boldsymbol{\theta}=(\boldsymbol{\theta}_\delta,\boldsymbol{\theta}_P,\boldsymbol{\theta}_\Gamma,\boldsymbol{\theta}_d)$ be the vector of parameters of the HSMM: $\boldsymbol{\theta}_\delta$ the initial probabilities, $\boldsymbol{\theta}_P$ the parameters of the state dependence distributions $f_i(x)$, $\boldsymbol{\theta}_\Gamma$ the conditional transition probabilities of the t.p.m, and $\boldsymbol{\theta}_d$ the parameters of the sojourn time distributions $d_i(u)$. 

\subsection{Autoregressive Structure (AR(p)-HMM and AR(p)-HSMM)}

In basic HMMs and HSMMs, we assume conditional independence of $\boldsymbol{X}$ given $\boldsymbol{C}$. However, for time series collected at fine temporal scales, the autocorrelation structure may not adequately be captured using this framework. A common extension is to assume that the state-dependent distribution of $\{X_t\}_{t=2}^T$ depends on the current state $C_t$ as well as on a subset of previous observations. This new assumption leads to an important class of models called Markov-switching models or autoregressive hidden (semi-)Markov models \citep{hamilton_1991}. 
In this framework, given the current state, $\{X_t\}$ is assumed to follow state-dependent Gaussian autoregressive processes of order $p$ $\left( \text{AR(p)} \right)$:
\begin{equation*}
x_t=\mu_i+\sum_{k=1}^p \omega_{ki}x_{t-k}+\epsilon_{ti} \quad \text{with } \epsilon_{ti}\sim N(0,\Sigma_i)
\end{equation*}
where $\mu_i$, $\omega_{k i}$ and $\Sigma_i$ are the conditional mean, autoregressive parameters and covariance matrix conditioned on $C_t=i$.  We denote the state-dependant distributions of an AR(p)-HSMM as
\begin{equation*}
f_i(x_t)=f(x_t|C_t=i,\boldsymbol{x}_{t-p:t-1})
\end{equation*}
and the joint distribution for vector $\boldsymbol{x}_{t:t'}$ with $t<t'$ as,
\begin{equation*}
f_i(\boldsymbol{x}_{t:t'})=
\prod_{k=t}^{t'} f(x_k|C_k=i,\boldsymbol{x}_{k-p:k-1})=
\prod_{k=t}^{t'}f_i(x_k).
\end{equation*}
The use and statistical properties of HMMs with autoregressive structure have been studied and formalized \citep{yang_properties_2000,francq_stationarity_2001}, and  further generalized for HSMMs \citep{xu_regularized_2020}. AR(p)-HSMM have been also used to perform classification in an unsupervised manner; \citep{bryan_autoregressive_2015} demonstrates its application to speech signal recognition and \citep{Duong} proposed the use of a switching hidden semi-Markov model to recognize human activities of daily living. However, to our knowledge, the use of AR(p)-HSMM to perform supervised classification has not been studied or applied.

\section{Classification via Supervised Learning}

In many applications, the primary objective of the analysis is to accurately predict the latent states of the system. The observations themselves are not of real interest as their function is merely to provide information about the states \citep{zucchini_hidden_2017}. For processes that evolve over time (or in sequence), HMMs (and their extensions) can be applied to classify observations into different latent states while taking into account that observations \textit{close to one another} in time are likely to have arisen from the same state. 
In this section we discuss the implementation of HMMs, HSMMs, and their respective autoregressive versions, in a supervised learning approach. 

Supervised learning methods assume that there is a data set for which labels are available for all samples. This labelled data set is used to fit the model and assess its predictive capability. Then, the labelled data set is distributed among the train set to fit the model and the test set to measure the accuracy of the final model. Finally, once the model has been evaluated, it can be used to classify unlabelled data.  Thus, when the aim is to perform classification using HMMs or HSMMs with a supervised approach, the steps are: split the labelled data, train the model (section \ref{inference}), predict hidden states (section \ref{prediction}), estimate the classification error (section \ref{error}), and eventually classify unlabelled data.

\subsection{Inference}\label{inference}

Inference for $\boldsymbol{\theta}$ is done via use of the complete-data likelihood as both the observations and values of the states are known in the training data. The complete-data likelihood is expressed as the joint distribution of the observations, $X_t$, and states, $C_t$,

\begin{align*}
f\left(\boldsymbol{c}_{1:T}, \boldsymbol{x}_{1:T}|\boldsymbol{\theta}\right) &=\operatorname{Pr}\left(\boldsymbol{C}_{1:T}=\boldsymbol{c}_{1:T}, \boldsymbol{X}_{1:T}=\boldsymbol{x}_{1:T} | \boldsymbol{\theta}\right) \\ &=\delta_{C_{1}} d_{C_{1}}\left(u_{1}\right) 
\prod_{t=1}^T f_{C_{t}}\left(x_t\right)\prod_{r \text{ is NAT}} \gamma_{c_{r-1}, c_{r}} d_{C_{r}}\left(u_{r}\right)  
\end{align*}

Note that for the autoregressive case $f_{C_t}(x_t)$ depends on the $p$ previous observations, i.e. $f_{C_t}(x_t)=f_{C_t}(x_t|\boldsymbol{x}_{t-p:t-1})$. The complete-data likelihood of an HMM necessarily has $d_{C_r}(.)\sim Geom(\lambda_r)$.

The complete-data likelihood can be expressed as a product of the terms of the t.p.m., initial distribution, state-dependent distributions and sojourn time distributions, which allows the complete-data log-likelihood to be written as a sum,
\begin{equation*}
log(f(\boldsymbol{c}_{1:T},\boldsymbol{x}_{1:T},\boldsymbol{\theta}))=log(\delta_{C_1})+log(d_{C_1}(u_1))+
\sum_{t=1}^Tlog(f_{C_t}(x_t))
+\sum_{\substack{t \in \text{ NAT}}} log(\gamma_{c_{r-1},c_r})+log(d_{C_r}(u_r))
\end{equation*}

For $K$ independent time series, the complete-data likelihood is the product of the individual complete-data likelihoods. As we conduct inference in a Bayesian framework, we first specify prior distributions for all parameters $f(\boldsymbol{\theta})$ and derive the joint posterior distribution $f(\boldsymbol{\theta}|\boldsymbol{x}_{1:T},\boldsymbol{C}_{1:T})\propto f(\boldsymbol{c}_{1:T},\boldsymbol{x}_{1:T},\boldsymbol{\theta})f(\boldsymbol{\theta})$. 
In particular, the joint log posterior distribution can be expressed as a sum of the components of the complete-data log likelihood and prior distributions. Full details of the prior distributions specified and derivation of the joint posterior distribution are provided in the Appendix (S1).

\subsection{State prediction}\label{prediction}

The objective of classification, in this context, is to determine which states are likely to have generated the observations. This process is known as \textit{state decoding} and can be done in one of two general ways: (i) \textit{local state decoding} -- identify the most likely state at each point in time or (ii) \textit{global state decoding} --  identify the sequence of states that is most likely to have generated the full observation sequence. In what follows we formalize and discuss the most common algorithms and their adaptations to HSMMs and AR-HSMMs for state decoding.   

\subsubsection{Local State Decoding }

Local state decoding is the process of determining at time $t$, what label of the state $C_t$ is most likely to have generated the observation $x_t$.
To do so, we compute and use the marginal probabilities $\operatorname{Pr}(C_t=i|\boldsymbol{X}_{1:T})$ via use of the forward-backward algorithm, which can be expressed as a function of the \textit{so-called} forward probabilities $\{\alpha_{t}\}_{t=1}^T$ and backward probabilities $\{\beta_t\}_{t=1}^T$. For the basic HMM, these quantities are given as 

$$\alpha_{t}(j)=\operatorname{Pr}\left(\boldsymbol{X}_{1:t}, C_{t}=j\right) \quad  \beta_{t}(j)=\operatorname{Pr}\left(\boldsymbol{X}_{t+1:T}| C_{t}=j\right)$$
where $\beta_T(j) = 1$, for all $j \in \{1, \ldots, J\}$ so that

$$\alpha_t(j) = f_j(x_t)\sum_i\gamma_{i,j}\alpha_{t-1}(i) \quad \beta_{t}(j) =\sum_i\beta_{t+1}(i)\gamma_{i,j}f_i(x_{t+1})$$

Given $\alpha_t(j)$ and $\beta_t(j)$, we can express the marginal probability as,
$$\operatorname{Pr}\left(C_{t}=j |\boldsymbol{X}_{1:T}\right)=\frac{\alpha_{t}(j) \beta_{t}(j)}{\text{Pr}(\boldsymbol{X}_{1:T})}.$$
Further details are given in \cite{bishop_model-based_2012} and \cite{zucchini_hidden_2017}. Derivation of the marginal state probabilities for the HSMM are similar to the HMM, while further adaptations are considered for the AR(p)-HSMM. We first express the forward and backwards probabilities as follows:   

\begin{equation*}
\alpha_{t}(j) = \operatorname{Pr}\left( \boldsymbol{X}_{1:t},C_{t]}=j\right)    
\quad
\beta_{t}(j) = \operatorname{Pr}\left(\boldsymbol{X}_{t+1:T}|C_{t]}=j,\boldsymbol{X}_{t-p:t}\right)
\end{equation*}

It is also possible to give recurrence equations that express $\alpha_t(\cdot)$ in terms of $\alpha_{t-d}(\cdot)$ and $\beta_{t}(\cdot)$ in terms of $\beta_{t+d}(\cdot)$. For $t=2,\dots,T$  and $j=1,\dots,J$ it can be shown that,

\begin{center}
\begin{itemize}
\centering
    \item[] $\alpha_{t}(j)=\sum_{d \in \mathcal{D}}\sum_{\substack{i\neq j}} \alpha_{t-d}(i) \gamma_{ij}d_j(d)f_j(\boldsymbol{x}_{t-d+1:t})$
    \item[] $ \beta_{t}(j)=\sum_{d \in \mathcal{D}}\sum_{\substack{i\neq j}} \beta_{t+d}(i) \gamma_{ji}d_i(d)f_i(\boldsymbol{x}_{t+1:t+d})$
\end{itemize}
\end{center}
Full details provided in Appendix S2.1. 
Let the marginal probabilities be expressed as $\xi_{t}(j)=\operatorname{Pr}\left(C_{t}=j , \boldsymbol{X}_{1:T}\right)$ and let
$\beta^*_t(j)=\operatorname{Pr}(\boldsymbol{X}_{t+1:T}|C_{[t+1}=j)$ (or $\beta^*_t(j)=\operatorname{Pr}(\boldsymbol{X}_{t+1:T}|C_{[t+1}=j,\boldsymbol{X}_{t-p:t})$ for the AR(p)-HSMM case).
It follows from the definition of $\beta^*_t(j)$ that,
\begin{equation*}
\begin{aligned}
 &\beta_t(j)=\displaystyle\sum_{i\neq j}\gamma_{ji}\beta^*_t(j)\\
&\beta_t^*(j)=\displaystyle\sum_{d \in \mathcal{D}} d_j(d)\beta_{t+d}(j)f_j(\boldsymbol{x}_{t+1:t+d})\\
\end{aligned}
\end{equation*}
We can then calculate the marginal probabilities as follows (full details provided in the Appendix S2.1). For $t=2,\dots,T$  and $j=1,\dots,J$
\begin{equation*}
\xi_{t}(j)=\xi_{t+1}(j)+\alpha_{t}(j)\sum_{i\neq j}\gamma_{ji}\beta_t^*(i)-\beta_t^*(j)\sum_{i\neq j}\alpha_t(i)\gamma_{ij}
\end{equation*}
For local state assignment, we use 
$\hat{C}_{t}=\arg \max _{j \in 1,\dots,J}\left\{\xi_{t}(j)\right\}$.

\subsubsection{Global State Decoding}

Global state decoding is a manner to identify the most likely sequence of states that could have given rise to the observed time series data. 
For this task, we determine the sequence  $c_1, c_2, \dots , c_T$ that maximizes the conditional
probability $\text{Pr}(\boldsymbol{C}_{1:T}|\boldsymbol{X}_{1:T})$ or equivalently the joint probability $\text{Pr}(\boldsymbol{C}_{1:T},\boldsymbol{X}_{1:T})$ via use of the Viterbi algorithm\citep{viterbi}. For the basic HMM we first define
$$\psi_{1 i}=\operatorname{Pr}\left(C_{1}=i, X_{1}\right)=\delta_{i} f_{i}\left(x_{1}\right)$$
and for $t=2,\dots,T,$
$$\psi_{t i}=\max _{c_{1:t-1}} \operatorname{Pr}\left(\boldsymbol{C}_{1:t-1}, C_{t}=i, \boldsymbol{X}_{1:t}\right)$$
Given the following recursion,
\begin{equation}
\psi_{t j}=\left(\max _{i}\left(\psi_{t-1, i} \gamma_{i j}\right)\right) f_{j}\left(x_{t}\right)
\label{recHMMViterbi}
\end{equation}
the most probable state sequence is obtained by

$$\hat{i}_{T}=\underset{i=1, \ldots, m}{\operatorname{argmax} \psi_{T i}}$$
and for $t=T-1,T-2,\dots, 1$
$$\hat{i}_{t}=\underset{i=1, \ldots, m}{\operatorname{argmax}}\left(\psi_{t i} \gamma_{i, i_{t+1}}\right)$$

Global state decoding for an HSMM is similar to that of a basic HMM, but also requires that most likely sojourn times be determined in addition to the most likely state sequence.
We define 

$$\psi_t(j,d)=\max_{c_{1:t-d}}\operatorname{Pr}\left(\boldsymbol{C}_{1:t-d},\boldsymbol{C}_{[t-d+1:t]=j},\boldsymbol{X}_{1:T}\right)$$

Similar to the basic HMM, it is possible to obtain the following recursion (full details provided in Appendix S2.2). For $t=2,\dots,T$, $d=1,\dots, D$  and $j=1,\dots,J$
\begin{equation*}
\psi_{t}(j,d)=\max_{\substack{i \neq j \\ d'\leq t} } \left\{\psi_{t-d}(i,d')\gamma_{ij}d_j(d) f_j(\boldsymbol{x}_{t-d+1:t})\right\}
\end{equation*}
Then, starting at time $T$, the most likely sequence of states and duration times is determined by 

$$(\hat{i}_T,\hat{d}_T)=\underset{\substack{j=1, \ldots, J \\ d=1, \dots D} }{\operatorname{argmax} \psi_{T}(j,d)}$$
and for $t=T-1,\dots 1$

$$(\hat{i}_t,\hat{d}_t)=\underset{\substack{j \neq {\hat{i}_{t+1}}\\d=1,\dots D}}{\operatorname{argmax}}\left(\psi_{t-d}(j,d') \gamma_{j, \hat{i}_{t+1}}d_{\hat{i}_{t+1}}(d)f_{\hat{i}_{t+1}}(\boldsymbol{x}_{t-d+1:t})\right)$$

The results of local and global state decoding are often very similar but not identical \citep{zucchini_hidden_2017}. One advantage of local state decoding via the forward-backward algorithm over global state decoding via the Viterbi algorithm is that it provides an estimated probability for every possible state in all times. These probabilities provide information about the uncertainty over the predictions which can be useful for model comparison.% and make decisions.   

\subsection{Classification Error}\label{error}

Given a fitted model, we can assess its capacity to classify the states correctly via both local and global state decoding. One of the most common measures to estimate prediction accuracy is the following classification accuracy index: the number of correct state predictions divided by the total number of observations in the time series. We can compute this measure from both the local and global decoding output. However, although this manner of estimating the prediction accuracy gives a general measure of the performance of the model, it does not take into account the uncertainty over the classification.  
Another manner to assess prediction capacity of our models is to use cross-entropy as it makes use of predicted probabilities of the states obtained via local state decoding and is given as,
$$CE=-\sum_{c=1}^{J} y_{x, c} \log \left(p_{x, c}\right)$$
where $y_{x,c}$ is a binary indicator of whether state label $c$ is the correct classification for observation $x$ ($y_{x,c}=1$ when the classification is correct), and $p_{x,c}$ is the  predicted probability observation that $x$ is generated by state $c$. The cross-entropy index increases as the predicted probability diverges from the actual label, such that a perfect model would have a CE value of $0$.

To measure out-of-sample prediction accuracy (or error), we adapt \textit{leave-one-out} cross validation to our application and use a \textit{leave-one time series-out} cross-validation approach \citep{Geisser,Refaeilzadeh2009}.  By doing so, we make the assumption that each time series is independent of each other and fulfill the requirement that the training and test sets are independent in order for cross-validation approaches to be valid.

\section{Simulation Study}
In some applications the assumption that the sojourn times follow geometric distributions (HMMs) is far from being realistic. In such cases, HSMMs may be considered more appropriate as they allow each state to have a variable duration time \citep{yu_hidden_2010}. 
However, when the main goal is to conduct classification it may not be necessary to extend the model to obtain correct state predictions. In this study we aim to characterize in which scenarios the state predictions of an HMM strongly differ from those of an HSMM, in order to understand when violating the assumption of geometric sojourn times can affect the accuracy of classifications.  

We simulate different scenarios of univariate time series of length 3000 from a 2-state HSMM with different sojourn times and observation distributions. Across all cases we assume Gaussian distributions for the observations and negative binomial distributions for the sojourn times. We propose three scenarios differing on the overlap between the observation distributions: the first assumes that $f_1(x_t)\sim N(0,1)$ and $f_2(x_t)\sim N(0.3,1)$ (high overlap), the second one assumes that $f_1(x_t)\sim N(0,1)$ and $f_2(x_t)\sim N(1,1)$ (medium overlap) and the last one assumes that $f_1(x_t)\sim N(0,1)$ and $f_2(x_t)\sim N(3,1)$ (low overlap). The overlap degree gives an idea of how similar are the observation distributions. A high overlap implies that the distributions are almost equal, i.e the area between them includes at least $80\%$ of them. A medium overlap indicates that the distributions are similar but they differ in at least a $40\%$, and finally a low overlap indicates that the distributions are different, sharing at most a $15\%$ of them. For each scenario we fix different parameter values for the negative binomial distributions: 
$d_1\sim NB(m_1,k_1)$ and $d_2\sim NB(m_2,k_2)$, where $m_1$ and $m_2$ are the means and $k_1$ and $k_2$ the dispersion parameters. When $k=1$ the negative binomial is a geometric distribution with probability $1/(m+1)$, and as $k$ takes values greater than 1 the dispersion of the distribution increases, moving away from the geometric distribution.  In order to characterize different scenarios, we first set an average and a difference between $m_1$ and $m_2$ ($(m_1+m_2)/2$ and $m_1-m_2$). We consider as average between $m_1$ and $m_2$: 90, 40 and 20, and as difference between them: 3, 15 and 30.  We then vary $k_1$ and $k_2$ so that either $k_1$ equals 1 (one is geometric: $k_1=1$ and $k_2=10$ or $k_2=30$) or both $k_1$ and $k_2$ are greater than one (none are geometric: $k_1=30$ and $k_2=50$ or $k_1=80$ and $k_2=100$). The values of $k_i$ and $m_i$ ($i=1,2$) were selected seeking to consider as many schemes as possible. When $k_i=1$  the most probable values is $1$ and as the value of $m_i$ increases the distributions have longer tails. When $k_i=10$ or $k_i=30$ the most probable value is nearer $m_i$, obtaining more concentrated distributions for $k_i=30$ and more dispersed distributions for $k_i=10$. The difference between $m_1$ and $m_2$ indicates how similar are the means of sojourn times of the two states.
For each parameter combination, we simulate ten time series of length 3000. To assess and compare the classification accuracy under both models (HMM and HSMM) we use \textit{leave-one time series-out} cross-validation approach as follows: we fit both models with all the simulated time series but one, then sampling from the joint posterior distribution obtained, we predict 30 times the hidden states of the simulated time series that we left out using Viterbi and FB algorithms. Finally we calculate the accuracy index of both predictions and the cross entropy value from the local decoding output. To fit the model we consider uninformative priors: $N(0,5)$ and $TN(0,5)$ for the mean and standard deviation of the observation distributions and $TN(20,50)$ for the mean and dispersion parameter of the sojourn time distributions. To perform the analysis we used R \citep{RCRAN}, to conduct parameter estimation we used Stan \citep{stan} and to implement the Viterbi algorithm for the HSMM we used the ``hsmm.viterbi" function of the R package ``hsmm'' (\url{https://cran.r-project.org/web/packages/hsmm/index.html}). The code is available at https://github.com/sofiar/HMM-HSMM-classification. 

Figure \ref{allsimuResults} shows the estimated accuracy values and the cross entropy index for each scenario obtained from the simulation analysis. It is clear that the performance of the models is lower when the overlap between the observations is larger: the cross entropy values increase and the accuracy indices decrease. When the average between $m_1$ and $m_2$ (the means of the sojourn time distributions) is large the differences between the three classification indices are low. However, when this average decreases the HSMM outperforms the HMM, since lower values of cross entropy index and higher values of accuracy are obtained. In general, for both models, the three indices improve as the average between $m_1$ and $m_2$ increases. With regards to the values of the dispersion parameters ($k_1$ and $k_2$) of the sojourn times, it is clear that when one of the distributions is geometric (case 1 in the axis label), the differences between the models are smaller. This last fact makes sense with the model hypothesis, i.e. in HMMs the sojourn times follow a geometric distribution and in the HSMM they follow a negative binomial distribution. 

The accuracy values of the FB algorithm and the Viterbi algorithm resulted to be practically equal in all the cases except when there is high overlap between the observation distributions and the average between $m_1$ and $m_2$ is low. In that case the accuracy values for the HSMM present higher dispersion: interquartile range between 0.05 and 0.55 for the HSMM, and 0.04 and 0.004 for the HMM. For the majority of the scenarios in terms of classification capacity the HSMM outperformed the HMM, obtaining larger differences when the overlap of the observation distributions increase, the average between the mean parameters of the sojourn times distributions decrease, and when the dispersion parameters are greater than one. Only when the overlap between the observation distributions is low and the averages between $m_1$ and $m_2$ are not too small, both models behave practically equally. 

\begin{figure}[hp!]
\centering
\includegraphics[scale = 0.3]{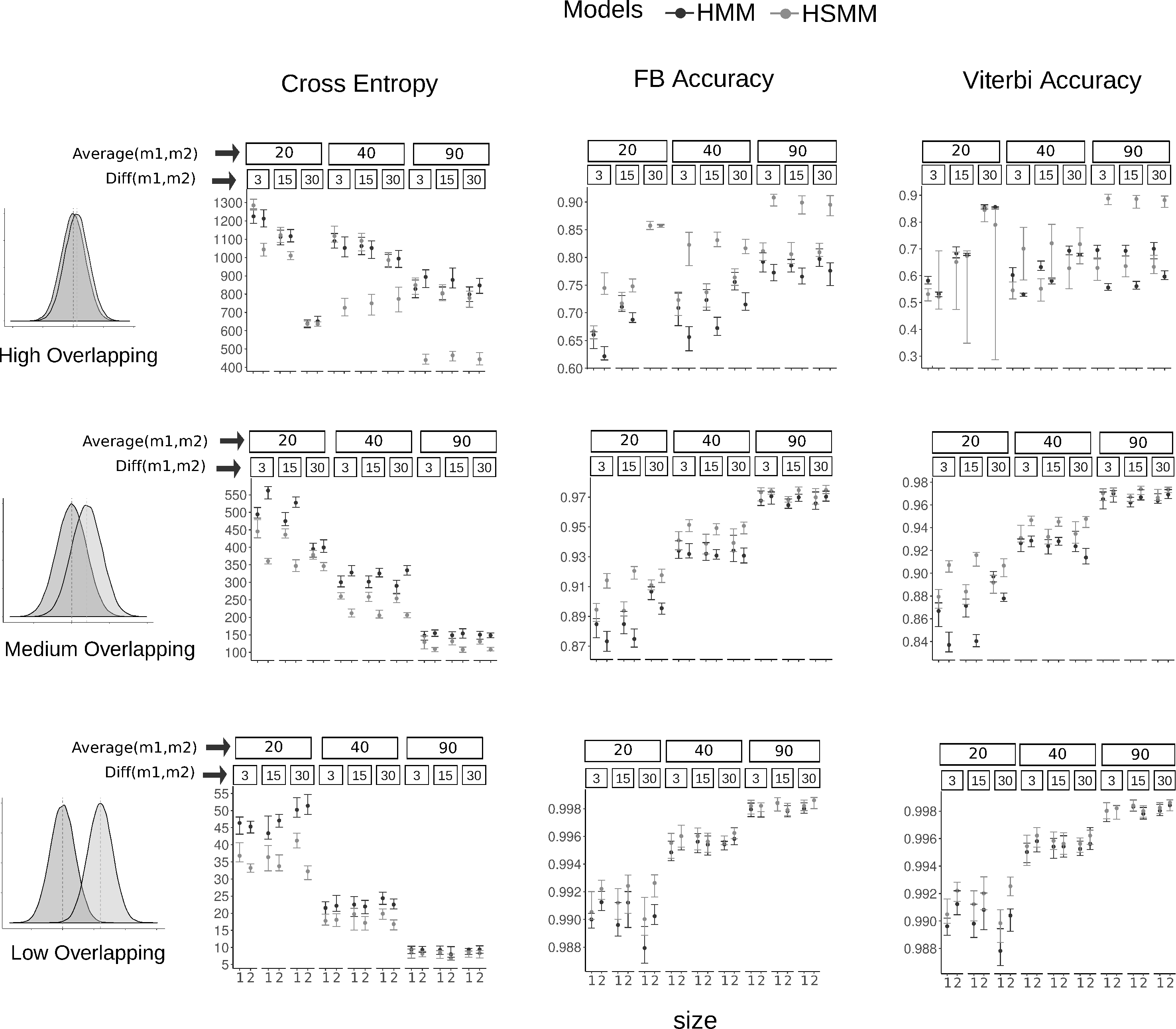}
\caption{Values of the three indices computed for the simulation study. By column the index calculated (Cross Entropy, Forward-Backward Accuracy and Viterbi Accuracy), by row the degree of overlap between the observation distributions (high, medium and low). In each box, the x-axis refers to the configuration of the dispersion parameters of the sojourn times ($k_1$ and $k_2$): for label $1$ $k_1=1$ (one is geometric) and for label $2$ neither $k_1$ nor $k_2$ equals $1$ (none is geometric). Groups of three consecutive values of each box correspond to the different averages between $m_1$ and $m_2$ ($20$,$40$ and $90$), and for each of them in order the three possible difference between $m_1$ and $m_2$ ($3$,$15$ and $30$). Points represent the median value and bars depict the first and the third quartile. With black the values obtained for the HMMs and with gray the values obtained for the HSMMs.  
} 
\label{allsimuResults}
\end{figure}

\section{Sheep acceleration Data}

We aim to classify accelerometer data from domestic merino sheep to distinguish between four different behavioral states.  As animal accelerometer data is typically collected at a fine temporal resolution, leading to high temporal dependence of the observation process, we proposed the use of HMMs and HSMMs to classify the states. Fieldwork was conducted at the "Campo Anexo Pilcaniyeu" from INTA (National Institute of Agricultural Technology) Bariloche, Patagonia Argentina. The data were collected from 25 different sheep during September and December of 2019. The animals were equipped with collars containing a DailyDiary\citep{wilson2008} developed at Swansea University, which were programmed to record 40 acceleration data per second (frequency of 40 Hz) and 13 magnetometer data per second (frequency of 13 Hz). The units where fixed on the top of the animal's neck and recorded acceleration in three axes: anterior–posterior (surge), lateral (sway) and dorso-ventral (heave).

In order to obtain the labelled data set, recording sessions were done on September 24th, 2019, and December 17th, 2019. These recordings served as the groundtruthing for our behavior recognition scheme. The acceleration data obtained was smoothed considering the moving average of windows of ten observations (a quarter of a second).
From the observed data we distinguished four behaviors: (i)Inactive (when the animal is still: resting or vigilant), (ii) Walk (normal walking speed with the head raised), (iii) Fast Walk (when the animal runs or moves fast) and (iv) Foraging (when the animal eats or looks for food. This can involve some walking but it is slow and with the head down). Manual classification was made at a 1s temporal resolution.

To define the variables involved in the observational process of the models, we identify the best characteristics that differentiate each behavior, i.e. we determine acceleration-derived features with low overlapping distributions between the states. To accomplish that, we analyzed the acceleration axes and two derived values: 
the vectorial dynamic body acceleration(VeDBA, \citep{Qasem}) and the head pitch angle \citep{wilson2008}. The VeDBA is the square root of the sum of the squares of the three acceleration axes; this quantity can be considered as a proxy of the animal energy expenditure \citep{Qasem}. The head pitch can be derived from an approximation of the running mean of the surge sensors and it is a proxy of the angle of the head. It is primarily used to determine if the animal's head is facing upwards, in a neutral position, or downwards. 

In the case of the sheep, the head pitch gives important information. The head of the animal looking down (negative pitch values) indicates foraging or searching behavior, the head in a neutral position (pitch values near zero) indicates that the sheep is probably in an inactive state, and finally a variable orientation of the head measured at a level of one or two seconds (high variance of the pitch values) implies that the the sheep could be running or walking.  The VeDBA index also gives useful information: low VeDBA values indicates less activity (inactive behavior) and high values indicates more activity (walk and fast walk behaviors). After analyzing several summary statistics derived from these quantities and seeking to differentiate as best as possible the pre-established behaviors, we selected three summary statistics over a one second window: \texttt{log(Mean.VeDBA)}, the logarithm of the mean VeDBA over one second; \texttt{Mean.Pitch}, the mean Pitch angle over one second; and the \texttt{log(Sd.Pitch)}, the logarithm of the standard deviation of the Pitch angle over one second.

To distinguish between the four behavioral states and considering a 1s temporal scale, we analyzed the classifications derived from four models: (i) HMM, (ii)AR(1)-HMM, (iii) HSMM and (iv)AR(1)-HSMM.   
The following assumptions were made: 
\begin{itemize}
\item[-] $\forall t=1,\dots,T$,  $C_t\in \{1,2,3,4\}$, with $1=\text{Walk}$, $2=\text{Fast-Walk}$, $3=\text{Inactive}$ and $4=\text{Foraging}$
\item[-] $\boldsymbol{\delta}=(\frac{1}{4},\frac{1}{4},\frac{1}{4},\frac{1}{4})$
\item[-] {$\forall t=1,\dots,T$, $\boldsymbol{X}_t=(X_t^1,X_t^2,X_t^3) $ with,
$X_t^1$ the \texttt{log(Mean.VeDBA)} , $X_t^2$ the \texttt{Mean.Pitch}, and $X_t^3$ the \texttt{log(Sd.Pitch)}}
\end{itemize}
For the non autoregressive models (i-iii)
\begin{itemize}
    \item[-] For $i \in \{1,2,3,4\}$ $P(\boldsymbol{X}_t|C_t=i) \sim N(\mu_i,\boldsymbol{\Sigma}_i)$, with $\boldsymbol{\Sigma}_i$ a diagonal matrix. 
\end{itemize}
For the autoregressive models (ii-iv)
\begin{itemize}
    \item[-] For $i \in \{1,2,3,4\}$ $P(\boldsymbol{X}_t|C_t=i) \sim N(\alpha_i+\beta_ix_{t-1},\boldsymbol{\Sigma_i})$, with $\boldsymbol{\Sigma}_i$ a diagonal matrix. 
\end{itemize}
and for the both HSMMs (iii-iv)
\begin{itemize}
    \item[-] For $i \in \{1,2,3,4\}$ $d_i\sim NB(m_i,k_i)$, where $m_i$ is the mean and $k_i$ the dispersion parameter. 
\end{itemize}
 
In order to assess the classification performance of the four models described above, we conducted a similar \textit{leave-one time series-out} cross-validation approach as detailed in the previous section.  To simplify the analysis, we considered the times series from each video as independent. However, as multiple videos correspond to the same animal it could be possible to add a random effect to take into account a potential correlation across time series from the same individual. In this case we first fit the four models with all the time series but one, and then predicted the remaining time series with the two decoding algorithms 100 times (sampling 100 times from the joint posterior distributions). Finally we calculated the two accuracy indices and the cross entropy value. Generally, when the aim is to identify different behaviors from acceleration data, one can incorporate previous knowledge about the mean, mode and/or variance into the distribution of the observations and sojourn times for each behavior. 
For instance, in the case of the sheep it is possible to assume that the sojourn time distribution of the resting class will contain larger values when compared to other sojourn time distributions, and that the distribution of mean pitch of the eating behavior should place almost all probability mass on negative values. As such, to fit the models we considered normal priors with high standard deviation and values of mean that differ according to each behavior. The list of the priors used are provided in the Appendix S3. 

Furthermore, we fit the four models with all the time series available to analyze the ability of them to capture the structure of the system.  In order to measure predictive capacity we computed a root mean square error (RMSE) over the observations: for each observed time series and model, using a 100 sample of the previously fitted posterior, we computed 100 predictions (using the FB algorithm) of the hidden states and the observation process. We then calculated the RMSE between the predictions and observations. To evaluate the goodness of fit we finally compared the estimations of the sojourn times assuming geometric distributions (HMM and AR(1)-HMM) and as negative binomial distributions (HSMM and AR(1)-HSMM) with the empirical distributions obtained from the observed sojourn times.

After discarding the videos that had not captured sheep or those for which acceleration data was not available, a total of 18 videos from eight different animals were obtained from the session of September and a total of 49 videos from 17 different animals from the session of December. Out of 67 time series with a total duration of two hours and 53 minutes, a $62.27\%$ were from the Inactive behavior, a $33.51\%$ from the Foraging state, a $3.75\%$ corresponded to the Walk class, and the remaining $0.47\%$ to the Fast Walk state. The data
used is available at https://doi.org/10.6084/m9.figshare.14473455.

Figure \ref{patternsexample} shows two examples of the signal patterns obtained for the four behaviors. When the sheep is inactive, the \texttt{log(Sd.Pitch)} and the \texttt{log(Mean.VedBA)} are low, meaning that the orientation of the head is almost constant and the energy expenditure is very low. When the sheep is eating or searching, the VeDBA values are higher and the pitch angle is lower. The patterns of the Walk behavior are similar to the previous one, however, the pitch angle values can be higher and less constant. The highest values of the standard deviation of the \texttt{log(Mean.Pitch)} and of the \texttt{log(Mean.VedBA)} are obtained when the sheep is running. It also can be noticed that the four behaviors have different duration times: the inactive state is the one that lasts the longest, followed by the Foraging state, then the Walk behavior and finally the Fast Walk state. 

Figure \ref{behadist} displays the boxplots obtained considering all the labelled data available for each behavior for the three variables considered. It can be observed that for the three variables the distributions are fairly different between the four states. Even though the values of \texttt{log(Mean.VedBA)} for the Walk and Foraging states are similar, for the Inactive behavior are clearly lower and for the Fast Walk behavior higher. In regards to the \texttt{Mean.Pitch}, the Fast Walk and Inactive behaviors show similar distributions, but for the Walk state the values of \texttt{Mean.Pitch} are lower and for the Foraging behavior are even minor (with dispersion). Finally, in spite of the fact that 
the values of \texttt{log(Sd.Pitch)} present greater dispersion, they are well differentiated between the four classes.

\begin{figure}
\centering
\includegraphics[scale = 0.5]{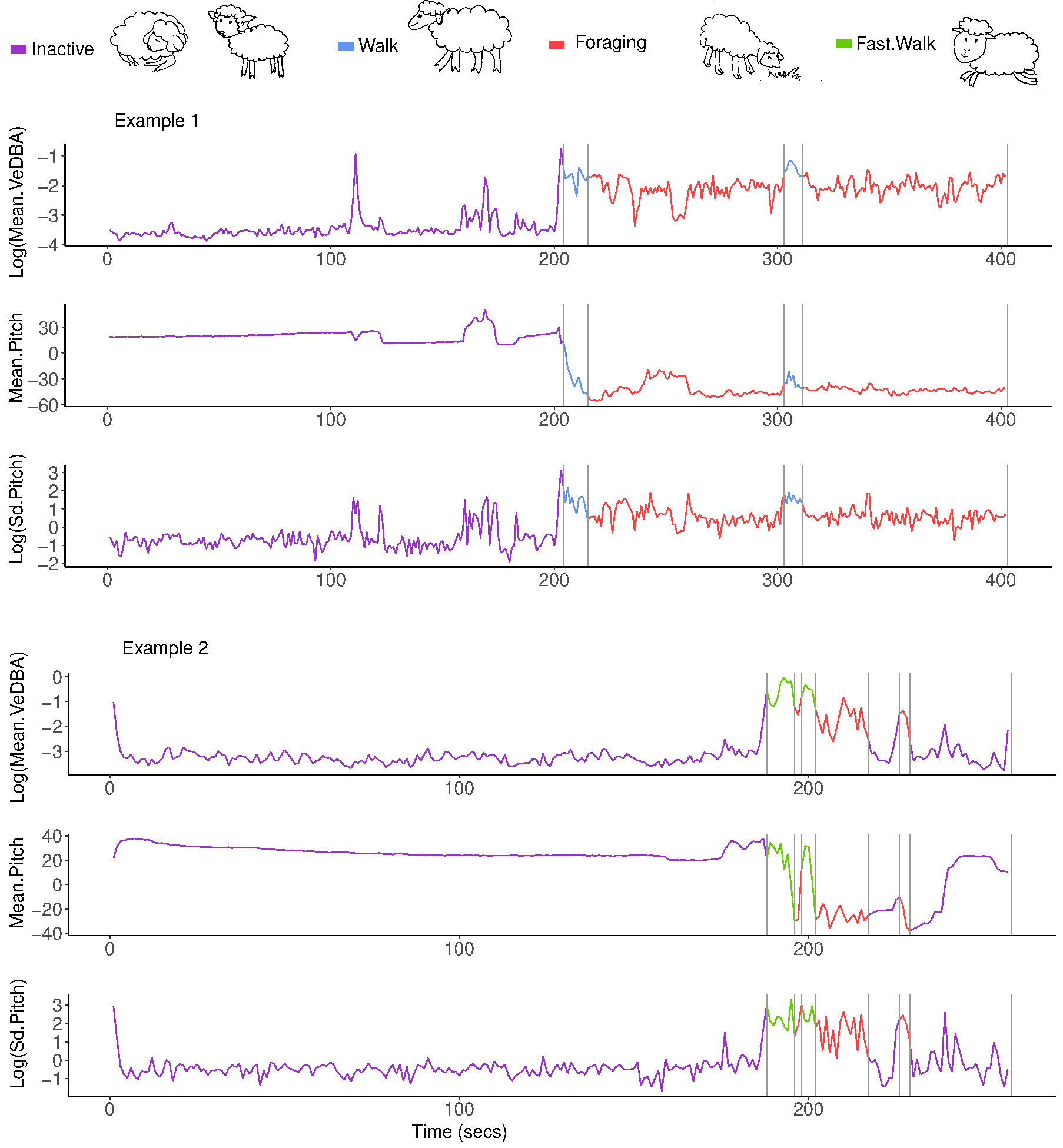}
\caption{Two examples of the three features calculated from the acceleration signal logged from one sheep: \texttt{log(Mean.VedBA)}, \texttt{Mean.Pitch} and \texttt{log(Sd.Pitch)}. Different colors represent the signals corresponding to each of the four behaviors: Inactive (rest or vigilance), Walk, Foraging and Fast Walk. }
\label{patternsexample}
\end{figure}

\begin{figure}
\centering
\includegraphics[scale = 0.9]{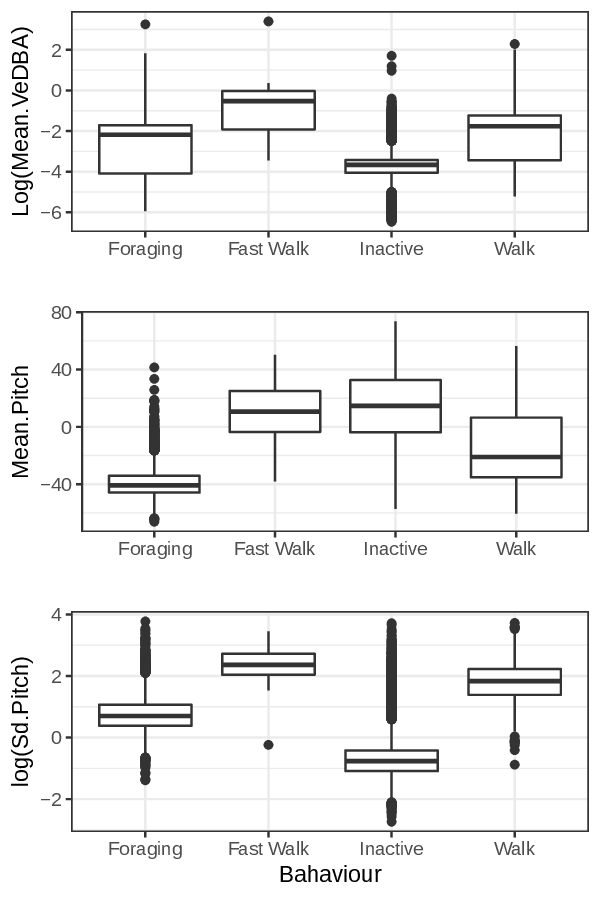}
\caption{Boxplots of the \texttt{log(Mean.VedBA)}, \texttt{Mean.Pitch} and \texttt{log(Sd.Pitch)} for each of the four behaviors:  Inactive (rest or vigilance), Walk, Foraging and Fast Walk.} 
\label{behadist}
\end{figure}

\begin{figure}[htp]
\centering
\includegraphics[scale = 0.8]{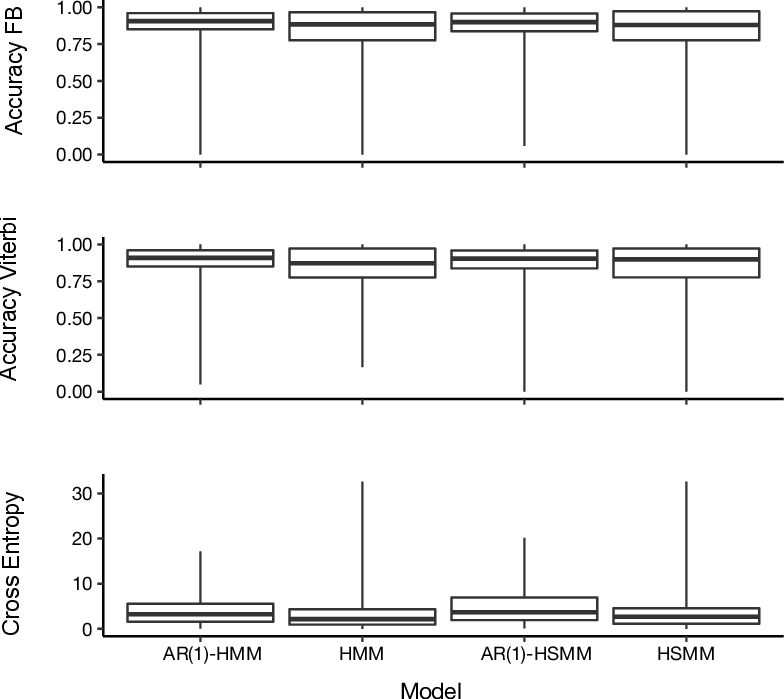}
\caption{Boxplots of the accuracy values of both decoding algorithms and Cross Entropy index obtained by the \textit{leave-one time series-out} cross-validation analysis for the four models considered: AR(1)-HMM, HMM, AR(1)-HSMM and HSMM.} 
\label{allResults}
\end{figure}

Figure \ref{allResults} shows the values of the three classification indices calculated for each model. In all cases the models behave in a similar way and predict correctly. However, the models with autoregressive structure (AR(1)-HMM and AR(1)-HSMM) show less dispersion in all indices, suggesting less uncertainty over the classifications. When comparing the performances of the FB algorithm and the Viterbi algorithm no significant differences can be appreciated in the values of the accuracy index, suggesting that global and local decoding perform similarly in terms of classifications.

\begin{figure}
\centering
\includegraphics[scale = 0.45]{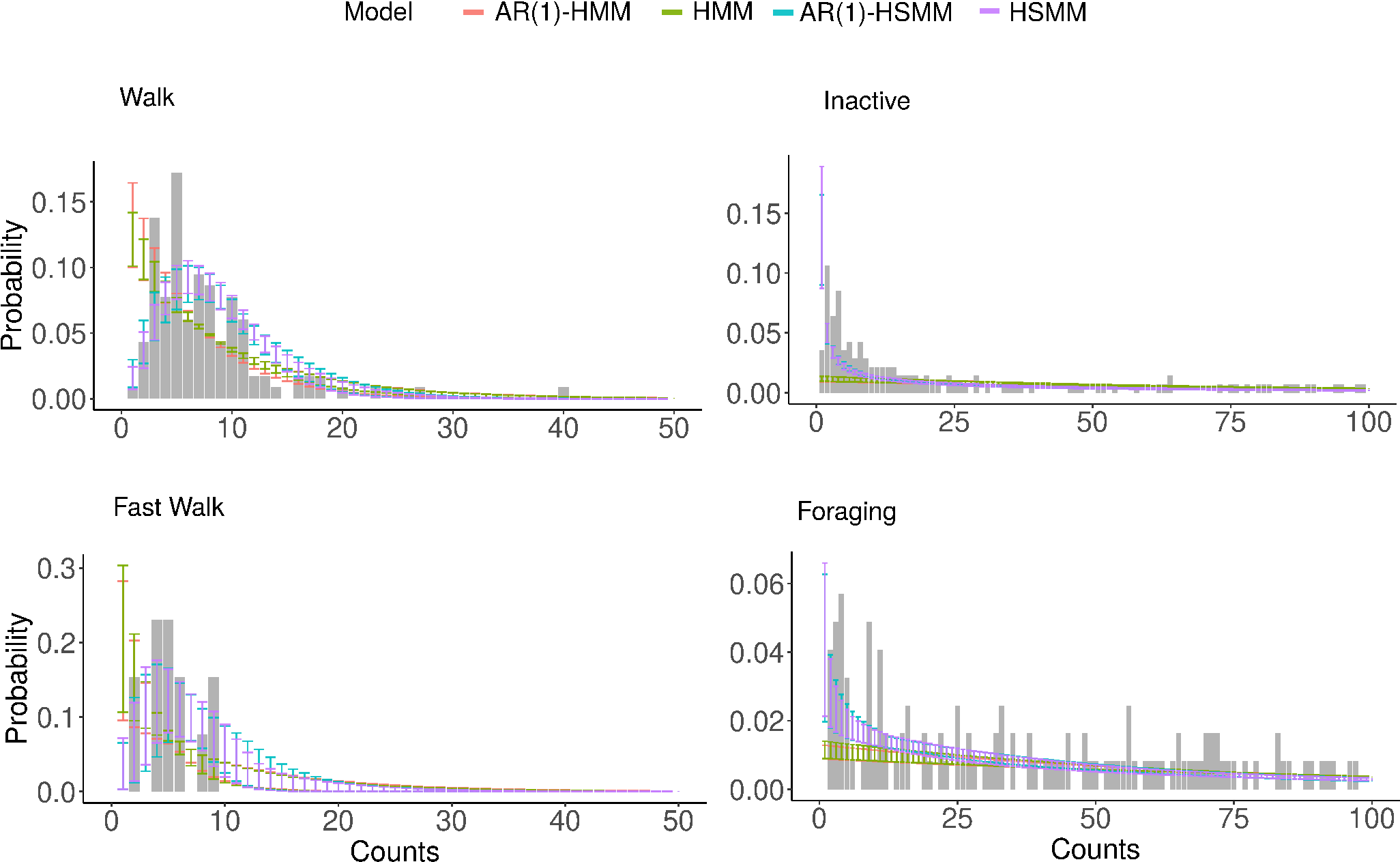}
\caption{Fitted distributions of the sojourn times for each behavior. The histogram displays the empirical distribution of the observed sojourn times while the four fitted models are displayed in color. The error bars indicate the pointwise estimates of the first and third quartile.} 
\label{Ndfs}
\end{figure}

The estimated RMSE values(given in Appendix S4) reveal that there are no important differences in the predictive capacity between the four models. The estimated RMSE behaves similarly across the three variables, except for the case of \texttt{Log(Mean.VeDBA)} for which the non autoregressive models (HMM and HSMM) show less uncertainty.
However, when analyzing the goodness of fit of the hidden process, it is clear that the estimations of the sojourn times of the HMMs differ from the estimations of the HSMMs (Figure \ref{Ndfs}). The empirical distribution of the sojourn times for the Walk and the Fast Walk behaviors evidence that the most likely dwell time is different from one, indicating that modeling these times with a geometric distribution would not be appropriate. Although in the case of the Inactive and the Foraging behaviors this pattern is less clear, the estimations of the HSMMs provide better approximations of the observed process. 

\section{Discussion}

We presented an overview of how to use HMMs and HSMMs with or without autoregressive structure in the observation process to perform supervised classification of time series data. We described how to extend an HMM to an HSMM, we detailed the model structures and presented the algorithms to predict the hidden states.
We then studied cases under which the classifications from both models applied to time series data truly differ.

When the final aim is to conduct state classification, from the simulation study we can conclude that when the state-dependent distributions have high overlap and the most likely dwell time is clearly different from one, 
unlike the geometric distribution, HSMMs outperform HMMs as the accuracy index is higher and the associated uncertainty is lower. However, when the overlap between the observation distributions is small or when the difference between the means of the sojourn time distributions is at least 25-35 units of time (regardless of their shape), the difference between models are minor. 

We applied two HMMs and two HSMMs to sheep accelerometer data obtaining accurate predictions from all of them: in all cases we obtained values of accuracy index greater than 0.87 on average. Nonetheless, we could not see important differences between the performance of the four models as the error measures were similar. According to the simulation study results, this could be explained mainly due to the low overlap between the observation distributions, although the estimations of the sojourn time distributions were more consistent with the semi-Markovian approach. 
However, due to the manner in which the data were collected, it was not possible to obtain enough information to distinguish between vigilance and resting thus having to consider both as just one behavior (Inactive). The acceleration data from these two behaviors is practically identical, however the times that sheep spend in these two behaviors are vastly different: long rest times and much shorter vigilance times. We posit that if the data set could be more complete, extending the recording sessions in order to observe enough resting samples, different performance should be achieved between the four models.  

When it comes to decide which model to use, the main consideration is the final goal of the analysis. When the aim is to conduct classification, the focus is on the ability to distinguish between different categories to then predict the unobserved states.
However, if the aim is to also provide a mechanism for data generation, the robustness of the model hypotheses becomes relevant. In those cases, correct specifications and accurate estimations should be taken into account to achieve the objectives of the analysis. 

When the aim is to select a model with good classification performance, various degrees of model misspecification may not result in state predictions that differ greatly from those obtained from the correctly specified model. For example, from our simulations it can be concluded that if the mean of the sojourn time distributions differ in more than 25-35 units of time, simple models incorrectly specified can perform equally as well as those that are correctly specified. However, if there is a high overlap between the state distributions, correctly specified models will improve the performance of the analysis.  

When the goal is to classify observations from time series data, HMMs and HSMMs resulted to be appropriate as they explicitly include sequential correlation into the analysis. Due to the structure of these models, they allow to distinguish between the sequential  dependence present in the data and the temporal structure of consecutive states. The difference between an HMM and an HSMM is that the latter allows for explicit duration modeling of the behavioral process. When there is significant overlap between the observation distributions and the sojourn times distributions do not have geometric form, considering extending the model to improve the model specifications should benefit the accuracy of the classifications. However, when the observation distributions are different between states and even more, if the sojourn time distributions also differ, even if the model fails in some assumptions, the classifications from an HMM can be as accurate as the classifications from an HSMM that is better specified. In these cases, we recommend the use of HMMs even if they are not correctly specified. As HSMMs are more complex models with a larger number of parameters than HMMs, they come at the cost of a more computationally expensive inferential process.

In practice when working with animal acceleration data observed labels may not be perfectly known. In certain cases, animals can be monitored in a laboratory setting \citep{wilson2008}, reducing the possibility of noisy labels. However, in other cases when the animal being studied cannot be observed in a controlled environment, it is not easy to collect a robust labelled set. The presence of noisy labels can affect the classification performance leading to lower prediction accuracy \citep{frenay}. Many studies have been devoted to the study of label noise and the development of techniques to deal with it \citep{Saez,GARCIA2015108,wang2021fair}, nevertheless until now there are no studies that consider this point to conduct classification of acceleration data. The inclusion of measurement error in the state process could be an interesting point for future research in the area of supervised classification with HMMs and HSMMs. 

In addition, more often than not, labelled data will not be available. Although the use of HMMs and HSMMs under the unsupervised approach is typically applied for different purposes than in classification, they can be equally useful. They can be used to discover groups of similar patters in the data and to have an approximate representation of the data generating process. When trying to identify behaviors from acceleration data without a training set, as multiple movement modes can correspond to the same behavior, the estimated states may not map to specific behaviors \citep{leos_barajas_analysis_2017}. This last fact should be considered when interpreting the predicted states. 
The limitation in such unsupervised learning settings, is to get a clear measure of the classification performance, as there is no manner to validate the state labels \citep{trevor_hastie_elements_nodate}. The implementation of HMMs and their extensions within an unsupervised learning framework can be particularly interesting as a direction for future research.

\bibliography{referenciasHMM.bib}

\end{document}

\maketitle

\section{Bayesian Inference}

To conduct classification via supervised learning, both the observations and values of the states are known from the training data. To fit the values of the model under the Bayesian approach, the posterior distributions of the parameters are computed. In section 3.1 of the main manuscript, we present the joint posterior distribution of the parameters as the product of the complete data likelihood and the prior distributions 

$$f(\boldsymbol{\theta}|\boldsymbol{x}_{1:T},\boldsymbol{c}_{1:T})
\propto f(\boldsymbol{c}_{1:T},\boldsymbol{x}_{1:T},\theta)f(\theta)$$

Let us distinguish the different elements of the parameter vector. Be $\boldsymbol{\theta}=(\boldsymbol{\delta},\boldsymbol{\gamma},\boldsymbol{\theta}_{obs},\boldsymbol{\theta}_{d})$, where $\boldsymbol{\delta}$ is the parameter vector of the initial distributions, $\boldsymbol{\gamma}$ the parameter vector of the transition probability matrix, $\boldsymbol{\theta}_{obs}$ the parameter vector of the observation distributions and $\boldsymbol{\theta}_d$ the vector parameter of the sojourn time distributions. If we consider independent priors 

$$f(\boldsymbol{\theta})=f(\boldsymbol{\delta})f(\boldsymbol{\gamma})f(\boldsymbol{\theta}_{obs})f\boldsymbol{(\theta}_d)$$

the posterior distribution can then be split into separate components

$$
\displaystyle f(\boldsymbol{\theta}|\boldsymbol{x}_{1:T},\boldsymbol{c}_{1:T}) = f(\boldsymbol{\delta}|\boldsymbol{x}_{1:T},\boldsymbol{c}_{1:T}) f(\boldsymbol{\gamma}|\boldsymbol{x}_{1:T},\boldsymbol{c}_{1:T})f(\boldsymbol{\theta}_{obs}|\boldsymbol{x}_{1:T},\boldsymbol{c}_{1:T})f(\boldsymbol{\theta}_d|\boldsymbol{x}_{1:T},\boldsymbol{c}_{1:T})         
$$

In what follows we detail the formulation of each component 

\subsection*{Specifications for the transition and initial distributions}

If we consider $f(\boldsymbol{\gamma})$ to folow a Dirichlet distribution, i.e.  $f(\boldsymbol{\gamma})\sim \mathcal{D}(\kappa_{j1},\dots,\kappa_{jJ})$ for $j=1\dots J$, as the changes of state have multinomial distribution, the posterior distribution for $\gamma$ is also Dirichlet

$$f(\boldsymbol{\gamma}|\boldsymbol{x}_{1:T},\boldsymbol{c}_{1:T})\propto \mathcal{D}(\kappa_{j1}+\nu_{j1},\dots,\kappa_{jJ}+\nu_{jJ})$$

with $\nu_{ji}=\text{number of occurrences from state  } j \text{ to } i$. Similarly for $f(\boldsymbol{\delta}|\boldsymbol{x}_{1:T},\boldsymbol{c}_{1:T})$, consider $f(\boldsymbol{\delta})\sim \mathcal{D}(\omega_{1},\dots,\omega_{J})$, then

$$f(\boldsymbol{\delta}|\boldsymbol{x}_{1:T},\boldsymbol{c}_{1:T})\propto \mathcal{D}(\omega_{1}+\tau_{1},\dots,\omega_{J}+\tau_{J})$$

with $\tau_{j}=\text{ number of times an observed series begins with state } j $

\subsection*{Specifications for the observation distributions}
Given $J$ states, it is necessary to make inference about the parameters of $J$ state-dependent distributions
$f_1(x)\ldots f_J(x)$. Each $j=1\dots J$ group of parameters, $\boldsymbol{\theta}_{obs_j}$, is estimated using only the observations allocated to state $j$. Then 

$$f_j(\boldsymbol{\theta}_{obs}|\boldsymbol{x}_{1:T} ,\boldsymbol{c}_{1:T})=f_j(\boldsymbol{\theta}_{obs_j}|\boldsymbol{x}_{[j]})$$

where $\boldsymbol{x}_{[j]}=\{x_t / C_t=j\}$.

%\{\text{observations indexed by } t \text{ such }C_t=J\}$ 

In this paper the state-dependent observations are assumed to be normally distributed

$$f_j(x)\sim N(\boldsymbol{\mu}_j,\boldsymbol{\Sigma}_j)$$

We consider independent Normal priors for the $J$ observation distributions. 
$f(\boldsymbol{\theta}_{obs_j})=f(\boldsymbol{\mu}_j)f(\boldsymbol{\Sigma}_j)$

$f(\boldsymbol{\mu}_j)\sim N(\boldsymbol{\mu}_{0j},\boldsymbol{\Sigma}_{0j})$ and $f(\boldsymbol{\Sigma}_j)\sim NT(\boldsymbol{\mu}_{0j},\boldsymbol{\Sigma}_{0j},0,\infty)$

where $NT$ indicates a truncated normal distribution. 
In this case, the posterior distributions do not have a closed form. They can be formulated as

$$f_j(\boldsymbol{\theta}_{obs_j}|\boldsymbol{x}_{1:T},\boldsymbol{c}_{1:T})\propto f(\boldsymbol{\theta}_{obs_j})\prod_ {t / C_t=j}f_j(x_t)$$

\subsection*{Specifications for the sojourn time distributions}

Once again, it is necessary to make inference about the parameters of $J$ state-dependent distributions 
$d_1(x).\dots d_J(x)$. Each $j=1\dots J$ group of parameters, $\boldsymbol{\theta}_{d_j}$, is estimated using only the 
the dwell times of the observations allocated to state $j$. 

%If the sojourn time distributions are assumed to be Poisson, $d_j\sim Poiss(\lambda_j)$, considering gamma priors  

As we consider negative binomial distributions for the $J$ sojourn time distributions, the posterior distributions do not have a closed form. They can be formulated as

$$f_j(\boldsymbol{\theta}_{d_j}|\boldsymbol{x}_{1:T},\boldsymbol{c}_{1:T})\propto f(\boldsymbol{\theta}_{d_j})\prod_{\substack{r \text{ is NAT}\\\text{and }C_r=j}}d_{Cr}(u_r)$$

\section{Proofs}
In the following section some of the results used in the main manuscript are proven. With \textcolor{blue}{blue} are marked the considerations for the autoregressive case.

\subsection{From Local Decoding section}
\begin{lema}
For $t=2,\dots,T$  and $j=1,\dots,J$
\begin{equation*}
  \alpha_{t}(j)=\sum_{d \in \mathcal{D}}\sum_{\substack{i\neq j}} \alpha_{t-d}(i) \gamma_{ij}d_j(d)f_j(\boldsymbol{x}_{t-d+1:t})
\end{equation*}
\end{lema}

\begin{proof}

\begin{flalign*}
\alpha_{t}(j)&=\operatorname{Pr}(C_{t]}=j, \boldsymbol{X}_{1:t})\\ 
&=\sum_{d \in \mathcal{D}} \operatorname{Pr}(\boldsymbol{C}_{[t-d+1:t]}=j, \boldsymbol{X}_{1:t})\\
&=\sum_{d \in \mathcal{D}} \operatorname{Pr}(\boldsymbol{X}_{t-d+1:t}|\boldsymbol{X}_{1:t-d},\boldsymbol{C}_{[t-d+1:t]}=j)\operatorname{Pr}(\boldsymbol{X}_{1:t-d},\boldsymbol{C}_{[t-d+1:t]}=j)\\
&=\sum_{d \in \mathcal{D}}f_j(\boldsymbol{x}_{t-d+1:t}) \operatorname{Pr}(\boldsymbol{X}_{1:t-d},\boldsymbol{C}_{[t-d+1:t]}=j)
\end{flalign*}

And we have 

\begin{flalign*}
\operatorname{Pr}(\boldsymbol{X}_{1:t-d},\boldsymbol{C}_{[t-d+1:t]}=j)&= \operatorname{Pr}(\boldsymbol{X}_{1:t-d}|\boldsymbol{C}_{[t-d+1:t]}=j)d_j(d)\\
&=\displaystyle\sum_{ i \neq j} \operatorname{Pr}(\boldsymbol{X}_{1:t-d},C_{t-d]}=i|\boldsymbol{C}_{[t-d+1:t]}=j)d_j(d)\\
&=\displaystyle\sum_{ i \neq j} \operatorname{Pr}(\boldsymbol{X}_{1:t-d},C_{t-d]}=i)\gamma_{ij}d_j(d)\\
&=\displaystyle\sum_{ i \neq j} \alpha_{t-d}(i)\gamma_{ij}d_j(d)\\
\end{flalign*}
\end{proof}

\begin{lema}
For $t=2,\dots,T$  and $j=1,\dots,J$
\begin{equation*}
  \beta_{t}(j)=\sum_{d \in \mathcal{D}}\sum_{\substack{i\neq j}} \beta_{t+d}(i) \gamma_{ji}d_i(d)f_i(\boldsymbol{x}_{t+1:t+d}) 
\end{equation*}
\end{lema}

\begin{proof}

\begin{equation*}
\begin{aligned}
\beta_{t}(j)&=\operatorname{Pr}(\boldsymbol{X}_{t+1:T}|C_{t]}=j,\boldsymbol{\textcolor{blue}{X_{t-p:t}}})\\
&=\sum_{d \in \mathcal{D}}\sum_{i \neq j}\operatorname{Pr}(\boldsymbol{X}_{t+1:T},\boldsymbol{C}_{[t+1:t+d]}=i|C_{t]}=j,\boldsymbol{\textcolor{blue}{X_{t-p:t}}})\\
&=\sum_{d \in \mathcal{D}}\sum_{i \neq j}\operatorname{Pr}(\boldsymbol{X}_{t+1:T}|\boldsymbol{C}_{[t+1:t+d]}=i,C_{t]}=j,\boldsymbol{\textcolor{blue}{X_{t-p:t}}})
\operatorname{Pr}(\boldsymbol{C}_{[t+1:t+d]}=i|C_{t]}=j,\boldsymbol{\textcolor{blue}{X_{t-p:t}}})\\
&=\sum_{d \in \mathcal{D}}\sum_{i \neq j}\operatorname{Pr}(\boldsymbol{X}_{t+1:t+d},\boldsymbol{X}_{t+d+1:T}|\boldsymbol{C}_{[t+1:t+d]}=i,C_{t]}=j,\boldsymbol{\textcolor{blue}{X_{t-p:t}}})\gamma_{ji}d_i(d)\\
&=\sum_{d \in \mathcal{D}}\sum_{i \neq j}\operatorname{Pr}(\boldsymbol{X}_{t+d+1:T}|\boldsymbol{X}_{t+1:t+d},\boldsymbol{C}_{[t+1:t+d]}=i,C_{t]}=j,\boldsymbol{\textcolor{blue}{X_{t-p:t}}})\operatorname{Pr}(\boldsymbol{X}_{t+1:t+d}|\boldsymbol{C}_{[t+1:t+d]}=i,C_{t]}=j,\boldsymbol{\textcolor{blue}{X_{t-p:t}}})
\gamma_{ji}d_i(d)\\
&=\sum_{d \in \mathcal{D}}\sum_{i \neq q}\operatorname{Pr}(\boldsymbol{X}_{t+d+1:T}|\boldsymbol{\mathsf{X}}_\mathsf{t+d+1-p:t+d},C_{t+d]}=i)\gamma_{ji}d_i(d)f_i(\boldsymbol{x}_{t+1:t+d}) 
\\
&=\sum_{d \in \mathcal{D}}\sum_{i \neq j}\beta_{t+d}(i)
\gamma_{ji}d_i(d)f_i(\boldsymbol{x}_{t+1:t+d})\\
\end{aligned}
\end{equation*}
\end{proof}

\begin{lema}
Given $\beta^*_q(j)=\operatorname{Pr}(\boldsymbol{X}_{t+1:T}|C_{[t+1}=j,\boldsymbol{\mathsf{X}}_\mathsf{t-p:t})$, for $t=2,\dots,T$ and $j=1,\dots,J$
$$\beta_t^*(j)=\displaystyle\sum_{d in \mathcal{D}} d_j(d)\beta_{t+d}(j)f_j(\boldsymbol{x}_{t+1:t+d})$$
\end{lema}
\begin{proof}
\begin{equation*}
\begin{aligned}
\beta_t^*(j)&=\operatorname{Pr}(\boldsymbol{X}_{t+1:T}|C_{[t+1}=j,\boldsymbol{\textcolor{blue}{X_{t-p:t}}})\\
&=\sum_{d \in \mathcal{D}}\operatorname{Pr}(\boldsymbol{X}_{t+1:T},\boldsymbol{C}_{[t+2:t+d]}=j|C_{[t+1}=j,\boldsymbol{\textcolor{blue}{X_{t-p:t}}})\\
&=\sum_{d \in \mathcal{D}}\operatorname{Pr}(\boldsymbol{X}_{t+1:T}|,\boldsymbol{C}_{[t+1:t+d]}=j,\boldsymbol{\mathsf{X}}_\mathsf{t-p:t})\operatorname{Pr}(\boldsymbol{C}_{[t+2:t+d]}=j|,C_{[t+1}=j,\boldsymbol{\textcolor{blue}{X_{t-p:t}}})\\
&=\sum_{d \in \mathcal{D}}\operatorname{Pr}(\boldsymbol{X}_{t+1:t+d},\boldsymbol{X}_{t+d+1:T}|,\boldsymbol{C}_{[t+1:t+d]}=j,\boldsymbol{\textcolor{blue}{X_{t-p:t}}})d_j(d)\\
&=\sum_{d \in \mathcal{D}}\operatorname{Pr}(\boldsymbol{X}_{t+d+1:T}|\boldsymbol{C}_{[t+1:t+d]}=j,\boldsymbol{\textcolor{blue}{X_{t-p:t+p}}})\operatorname{Pr}(\boldsymbol{X}_{t+1:t+d}|\boldsymbol{C}_{[t+1:t+d]}=j,\boldsymbol{\textcolor{blue}{X_{t-p:t}}})
d_j(d)\\
&=\sum_{d \in \mathcal{D}}\operatorname{Pr}(\boldsymbol{X}_{t+d+1:T}|C_{t+d]}=j,\boldsymbol{\textcolor{blue}{X_{t+1+d-p:t+d}}})d_j(d)f_j(\boldsymbol{x}_{t+1:t+d})\\
&=\sum_{d \in \mathcal{D}}\beta_{t+d}(j)d_j(d)f_j(\boldsymbol{x}_{t+1:t+d})\\
\end{aligned}
\end{equation*}
\end{proof}

\begin{lema}
Given the definitions of $\xi_{t}(j)$, $\alpha_{t}(j)$ and $\beta_t^*(j)$, for $t=2,\dots,T$  and $j=1,\dots,J$
\begin{equation*}
\xi_{t}(j)=\xi_{t+1}(j)+\alpha_{t}(j)\sum_{i\neq j}\gamma_{ji}\beta_t^*(i)-\beta_t^*(j)\sum_{i\neq j}\alpha_t(i)\gamma_{ij}
\end{equation*}
\label{ld3}
\end{lema}

\begin{proof}
We need the following equations for the proof 
\begin{itemize}
    \item $\operatorname{Pr}(C_t=j,\boldsymbol{X}_{1:T})=\operatorname{Pr}(C_{t+1}=j,\boldsymbol{X}_{1:T})+\operatorname{Pr}(C_{t]}=j,\boldsymbol{X}_{1:T})-\operatorname{Pr}(C_{[t+1}=j,\boldsymbol{X}_{1:T})$
    \item $\operatorname{Pr}(C_{t]}=i, \boldsymbol{C}_{[t+1:t+d]}=j,\boldsymbol{X}_{1:T})=\alpha_t(i)\gamma_{ij}d_j(d)f_j(\boldsymbol{x}_{t+1:t+d})\beta_{t+d}(j)$ (Lemma \ref{lem1})
\end{itemize}

Then

\begin{equation*}
\begin{aligned}
\xi_{t}(j)&=\operatorname{Pr}(C_t=j,\boldsymbol{X}_{1:T})\\
&=\operatorname{Pr}(C_{t+1}=j,\boldsymbol{X}_{1:T})+\operatorname{Pr}(C_{t]}=j,\boldsymbol{X}_{1:T})-\operatorname{Pr}(C_{[t+1}=j,\boldsymbol{X}_{1:T})\\
&=\xi_{t+1}(j)+\operatorname{Pr}(C_{t]}=j,\boldsymbol{X}_{1:T})-\operatorname{Pr}(C_{[t+1}=j,\boldsymbol{X}_{1:T})\\ 
&=\xi_{t+1}(j)+\sum_{i\neq j}\operatorname{Pr}(C_{t]}=j,C_{[t+1}=i,\boldsymbol{X}_{1:T})-\sum_{i\neq j}\operatorname{Pr}(C_{t]}=i,C_{[t+1}=j,\boldsymbol{X}_{1:T})\\
&=\xi_{t+1}(j)+\sum_{i\neq j}\sum_{d \in D}\operatorname{Pr}(C_{t]}=j,\boldsymbol{C}_{[t+1:t+d]}=i,\boldsymbol{X}_{1:T})
-\sum_{i\neq j}\sum_{d \in D}\operatorname{Pr}(C_{t]}=i,\boldsymbol{C}_{[t+1:t+d]}=j,\boldsymbol{X}_{1:T})\\
&=\xi_{t+1}(j)+\sum_{i\neq j}\sum_{d \in D} \alpha_t(j)\gamma_{ji}d_j(d)f_j(\boldsymbol{x}_{t+1:t+d})\beta_{t+d}(i)-\sum_{i\neq j}\sum_{d \in D} \alpha_t(i)\gamma_{ij}d_i(d)f_i(\boldsymbol{x}_{t+1:t+d})\beta_{t+d}(j)\\
% Reordeno
&=\xi_{t+1}(j)+\alpha_t(j)\sum_{i\neq j}\gamma_{ji} \sum_{d \in D} d_j(d)f_j(\boldsymbol{x}_{t+1:t+d})\beta_{t+d}(i) -\sum_{i\neq j}\alpha_t(i)\gamma_{ij}\sum_{d \in D} d_i(d)f_i(\boldsymbol{x}_{t+1:t+d})\beta_{t+d}(j)\\ 
\end{aligned}
\end{equation*}

and finally, using that $\displaystyle\beta_t(j)=\sum_{i\neq j}\gamma_{ij}\beta_t^*(j)$,  we have 

$
\xi_{t}(j)=\xi_{t+1}(j)+\alpha_{t}(j)\displaystyle\sum_{i\neq j}\gamma_{ji}\beta_t^*(i)-\beta_t^*(j)\sum_{i\neq j}\alpha_t(i)\gamma_{ij}
$

\end{proof}

The following lemma was necessary to prove the previous proposition \ref{ld3}
\begin{lem}
$$\operatorname{Pr}(C_{t]}=i, \boldsymbol{C}_{[t+1:t+d]}=j,\boldsymbol{X}_{1:T})=\alpha_t(i)\gamma_{ij}d_j(d)f_j(\boldsymbol{x}_{t+1:t+d})\beta_{t+d}(j)$$
\label{lem1}
\end{lem}
\begin{proof}
\begin{equation*}
\begin{aligned}
\operatorname{Pr}(C_{t]}=i, \boldsymbol{C}_{[t+1:t+d]}=j,\boldsymbol{X}_{1:T})&= \operatorname{Pr}(C_{t]}=i, \boldsymbol{C}_{[t+1:t+d]}=j,\boldsymbol{X}_{1:t},\boldsymbol{X}_{t+1:T})\\
&=\operatorname{Pr}(\boldsymbol{C}_{[t+1:t+d]}=j,\boldsymbol{X}_{t+1:T}|C_{t]}=i,\boldsymbol{X}_{1:t})\operatorname{Pr}(C_{t]}=i,\boldsymbol{X}_{1:t})\\
&=\operatorname{Pr}(\boldsymbol{C}_{[t+1:t+d]}=j,\boldsymbol{X}_{t+1:T}C_{t]}=i,\boldsymbol{X}_{1:t})\alpha_t(i)\\
&=\operatorname{Pr}(\boldsymbol{X}_{t+1:T}|C_{t]}=i,\boldsymbol{X}_{1:t},\boldsymbol{C}_{[t+1:t+d]}=j)\operatorname{Pr}(\boldsymbol{C}_{[t+1:t+d]}=j|C_{t]}=i,\boldsymbol{X}_{1:t})\alpha_t(i)\\
&=\operatorname{Pr}(\boldsymbol{X}_{t+1:t+d},\boldsymbol{X}_{t+d+1:T}|C_{t]}=i,\boldsymbol{X}_{1:t},\boldsymbol{C}_{[t+1:t+d]}=j)\gamma_{ij}d_j(d)\alpha_t(i)\\
&=\operatorname{Pr}(\boldsymbol{X}_{t+d+1:T}|C_{t]}=i,\boldsymbol{X}_{1:t},\boldsymbol{C}_{[t+1:t+d]}=j,\boldsymbol{X}_{t+1:t+d})\\
&\operatorname{Pr}(\boldsymbol{X}_{t+1:t+d}|C_{t]}=i,\boldsymbol{X}_{1:t},\boldsymbol{C}_{[t+1:t+d]}=j)
\gamma_{ij}d_j(d)\alpha_t(i)\\
&=\operatorname{Pr}(\boldsymbol{X}_{t+d+1:T}|C_{t+d]}=j
, \boldsymbol{\mathsf{X}}_\mathsf{t-p:t-1})
\operatorname{Pr}(\boldsymbol{X}_{t+1:t+d}|\boldsymbol{\mathsf{X}}_\mathsf{t+1-p:t+d-1})\gamma_{ij}d_j(d)\alpha_t(i)\\
&=\beta_{t+d}(j)f_j(\boldsymbol{x}_{t+1:t+d})\gamma_{ij}d_j(d)\alpha_t(i)
\end{aligned}
\end{equation*}
\end{proof}

\subsection{From Global Decoding section}

\begin{lema}
Be  
$$\psi_t(j,d)=\max_{c_{1:t-d}}\operatorname{Pr}\left(\boldsymbol{C}_{1:t-d},\boldsymbol{C}_{[t-d+1:t]=j},\boldsymbol{X}_{1:T}\right)$$
for $t=2,\dots,T$, $d=1,\dots, D$  and $j=1,\dots,J$, then 

\begin{equation*}
\psi_{t}(j,d)=\max_{\substack{i \neq j \\ d'\leq t} } \left\{\psi_{t-d}(i,d')\gamma_{ij}d_j(d) f_j(\boldsymbol{x}_{t-d+1:t})\right\}
\end{equation*}
\end{lema}

\begin{proof}

By definition we have 
\begin{equation*}
\begin{aligned}
\psi_{t-d}(i,d')=\max_{c_{1:t-d-d'}} \operatorname{Pr}(\boldsymbol{C}_{1:t-d-d'},\boldsymbol{C}_{[t-d-d'+1:t-d]}=i,\boldsymbol{X}_{1:t-d})
\end{aligned}
\end{equation*}
In the same way,
\begin{equation*}
\begin{aligned}
\psi_{t}(j,d)&=\max_{c_{1:t-d}} \operatorname{Pr}(\boldsymbol{C}_{1:t-d},\boldsymbol{C}_{[t-d+1:t]}=j,\boldsymbol{X}_{1:t})\\
&=\max_{\substack{c_{1:t-d-d'}\\c_{t-d'-d+1:t-d}\\d'\leq d}} \operatorname{Pr}(\boldsymbol{C}_{1:t-d'-d},\boldsymbol{C}_{t-d'-d+1:t-d},\boldsymbol{C}_{[t-d+1:t]}=j,\boldsymbol{X}_{1:t-d},\boldsymbol{X}_{t-d+1:t})\\
&=\max_{\substack{c_{1:t-d-d'}\\i\neq j\\d'\leq d}} \operatorname{Pr}(\boldsymbol{C}_{1:t-d'-d},\boldsymbol{C}_{[t-d'-d+1:t-d]}=i,\boldsymbol{C}_{[t-d:t]}=j,\boldsymbol{X}_{1:t-d},\boldsymbol{X}_{t-d+1:t})\\
&=\max_{\substack{c_{1:t-d-d'}\\i\neq j\\d'\leq d}} \operatorname{Pr}(\boldsymbol{C}_{1:t-d'-d},\boldsymbol{C}_{[t-d'-d+1:t-d]}=i,\boldsymbol{X}_{1:t-d})\gamma_{ij}d_j(d)\operatorname{Pr}(\boldsymbol{X}_{t-d+1:t}|\boldsymbol{\textcolor{blue}{X_{t-d+1-p:t+d}}})\\
&=\max_{\substack{c_{1:t-d-d'}\\i\neq j\\d'\leq d}} \operatorname{Pr}(\boldsymbol{C}_{1:t-d'-d},\boldsymbol{C}_{[t-d'-d+1:t-d]}=i,\boldsymbol{X}_{1:t-d})\gamma_{ij}d_j(d)f_j(\boldsymbol{x}_{t-d+1:t})\\
&=\max_{\substack{i\neq j\\d'\leq d}} \psi_{t-d}(i,d')\gamma_{i,j}d_j(d)f_j(\boldsymbol{x}_{t-d+1:t})
\end{aligned}
\end{equation*}

\end{proof}

\newpage

\section{Prior distributions}

In this section we show the priors used to conduct the simulation study (table \ref{priorsSimu}) and the sheep acceleration data analysis (table \ref{PriorsReal})

\begin{table}[H]
\begin{tabular}{p{1.5 cm}p{4 cm}p{10cm}}
Parameter & Prior Distribution  & Interpretation    \\
\cline{1-3} 
$\mu_i$ & Normal(0,5) & Mean for the Normal distribution describing the observation given state $i$. \\
$\sigma_i$ & Trunc normal(0,5) & Standard deviation for the Normal distribution describing the observation given state $i$.   \\
$m_i$ & Trunc. normal(20,50) & Mean for the Negative binomial distributions describing the sojourn time of each state.\\
$k_i$ & Trunc. normal(20,50) & Dispersion parameter for the Negative binomial distributions describing the sojourn time of each state.\\
\end{tabular}
\caption{Prior distributions used to conduct the simulation study}
\label{priorsSimu}
\end{table}

\begin{table}[H]
\begin{tabular}{p{1.5 cm}p{7 cm}p{7cm}}
Parameter & Prior Distribution  & Interpretation    \\
\cline{1-3} 
$\boldsymbol{\alpha}_1$ & $\text{MultiNormal}((0,20,0)^T, diag(3,5,10))$& Mean vector for the Normal distribution describing the observation for the walk behaviour \\
$\boldsymbol{\alpha}_2$ & $\text{MultiNormal}((0,20,0)^T, diag(3,5,10))$& Mean vector for the Normal distribution describing the observation for the fast walk behaviour \\
$\boldsymbol{\alpha}_3$ & $\text{MultiNormal}((0,20,0)^T, diag(3,5,10))$& Mean vector for the Normal distribution describing the observation for the inactive behaviour \\
$\boldsymbol{\alpha}_4$ & $\text{MultiNormal}((0,-20,0)^T, diag(3,5,10))$& Mean vector for the Normal distribution describing the observation for the foraging behaviour \\
$\boldsymbol{\sigma}_i$ & $\text{Trunc. MultiNormal}((0,0,0)^T, diag(3,10,3))$& Standard deviation vector for the Normal distribution describing the observation for each behaviour $i=1,2,3,4$\\
$m_1$ & Trunc. normal(0,5)& Mean for the Negative binomial distributions describing the sojourn time of the walk behaviour\\
$m_2$ & Trunc. normal(0,5)& Mean for the Negative binomial distributions describing the sojourn time of the fast walk behaviour\\
$m_3$ &Trunc. normal(80,30)& Mean for the Negative binomial distributions describing the sojourn time of the inactive behaviour\\
$m_4$ & Trunc. normal(25,20)& Mean for the Negative binomial distributions describing the sojourn time of the foraging behaviour\\
$k_i$ & Trunc. normal(0,5)& Dispersion parameter for the Negative binomial distributions describing the sojourn time of each behaviour $i=1,2,3,4$\\
$\beta_i$ & Normal(0,5)& Slope of the autorregresive term for each behaviour $i=1,2,3,4$\\

\end{tabular}
\caption{Prior distributions used to conduct the analysis of the sheep acceleration data}
\label{PriorsReal}
\end{table}

\section{RMSE extra figure}

%Some additional figures are display in the following section.

In order to measure the predictive capacity of the models proposed when analyse the accelerometer sheep data, the Root Mean Squared Error (RMSE) was calculated over the four model: AR(1)-HMM, HMM, AR(1)-HSMM and HSMM. For each observed time series and model, using a 100 sample of the fitted posterior, we computed 100 predictions (using the FB algorithm) of the hidden states and the observation process. Figure \ref{RMSE} shows the values obtained 

\begin{figure}[h]
\centering
\includegraphics[scale = 0.8]{RMSEgray.eps}
\caption{Boxplots of the Root Mean Squared Error over the observations for the four models considered.} 
\label{RMSE}
\end{figure}

%\bibliography{referenciasHMM.bib}